\documentclass[12pt]{iopart}

\usepackage{amssymb,amsthm,graphicx} 
\usepackage{hyperref}
\hypersetup{
	colorlinks,
	citecolor=blue,
	filecolor=black,
	linkcolor=blue,
	urlcolor=blue
}

\newcommand{\eqref}[1]{(\ref{#1})}

\newtheorem{theorem}{Theorem}
\newtheorem{proposition}{Proposition}
\newtheorem{corollary}{Corollary}
\newtheorem{lemma}{Lemma}

\theoremstyle{remark}
\newtheorem{remark}{Remark}

\newenvironment{subproof}[1][\proofname]{%
  \begin{proof}[#1]%
}{%
  \end{proof}%
}

\begin{document}

\title[Universality of the fully connected vertex in Laplacian CTQW problems]{Universality of the fully connected vertex in Laplacian continuous-time quantum walk problems}

\author{Luca Razzoli$^{1}$, Paolo Bordone$^{1,2}$, Matteo G. A. Paris$^{3,4}$}

\address{$^{1}$Dipartimento di Scienze Fisiche, Informatiche e Matematiche, Universit\`{a} di Modena e Reggio Emilia, I-41125 Modena, Italy}

\address{$^{2}$Centro S3, CNR-Istituto di Nanoscienze, I-41125 Modena, Italy}

\address{$^{3}$Quantum Technology Lab, Dipartimento di Fisica {\em Aldo Pontremoli}, Universit\`{a} degli Studi di Milano, I-20133 Milano, Italy}
\address{$^{4}$INFN, Sezione di Milano, I-20133 Milano, Italy}

\eads{\mailto{luca.razzoli@unimore.it},
\mailto{paolo.bordone@unimore.it}, and 
\mailto{matteo.paris@fisica.unimi.it}}

\vspace{10pt}
\begin{indented}
\item[]7 May 2022
\end{indented}

\begin{abstract}
A fully connected vertex $w$ in a simple graph $G$ of order $N$ is a vertex connected to all the other $N-1$ vertices. Upon denoting by $L$ the Laplacian matrix of the graph, we prove that the continuous-time quantum walk (CTQW)---with Hamiltonian $H=\gamma L$---of a walker initially localized at $\vert w \rangle$ does not depend on the graph $G$. We also prove that for any Grover-like CTQW---with Hamiltonian $H=\gamma L +\sum_w \lambda_w \vert w \rangle\langle w \vert$---the probability amplitude at the fully connected marked vertices $w$ does not depend on $G$. The result does not hold for CTQW with Hamiltonian $H=\gamma A$ (adjacency matrix). We apply our results to spatial search and quantum transport for single and multiple fully connected marked vertices, proving that CTQWs on any graph $G$ inherit the properties already known for the complete graph of the same order, including the optimality of the spatial search. Our results provide a unified framework for several partial results already reported in literature for fully connected vertices, such as the equivalence of CTQW and of spatial search for the central vertex of the star and wheel graph, and any vertex of the complete graph.
\end{abstract}

\vspace{2pc}
\noindent{\it Keywords}: quantum walks, quantum search, Grover search, quantum transport, Laplacian matrix, graphs
%
%
%

\section{Introduction}
A quantum particle propagating on a discrete space, e.g., on a graph, performs a quantum walk, the quantum analog of classical random walk. Quantum walks are a well-established model \cite{portugal2018quantum}, with already existing  physical implementations \cite{wang2013physical}. Continuous-time quantum walks (CTQWs) were introduced in \cite{farhi1998quantum} as a quantum algorithm to traverse decision trees.
In a CTQW the state of the walker evolves continuously in time according to the Schr\"{o}dinger equation under a Hamiltonian which respects the topology of the graph considered. The graph is mathematically represented by the Laplacian matrix $L = D-A$, which encodes the degree $D$ and the adjacency $A$ of the vertices. Hence, the matrices $L$ and $A$ are usually taken as generators of a CTQW. For regular graphs, $A$ and $L$ are equivalent, since all the vertices have the same degree and thus $D$ is proportional to the identity. For irregular graphs, instead, $A$ and $L$ are not equivalent in general, but it is possible to recover the same probability distributions for certain graphs and depending on the initial states \cite{wong2021equivalent}.

CTQWs walks inherit the versatility of application from their classical ancestors, but the peculiar features arising from their quantum nature---e.g., the superposition of the quantum walker in their path---make them suitable candidates not only for modeling physical processes, such as coherent transport in complex networks \cite{mulken2011ctqw_transport} even in biological system \cite{mohseni2008environment}, but also for applications in quantum technologies. Indeed, they are of use in studying perfect state transfer in quantum spin networks \cite{christandl2004perfect,alvir2016perfect}, which are of utmost importance for quantum communication, they can be used to develop quantum algorithms, such as spatial search \cite{childs2004spatial,wong2016laplacian,chakraborty2020optimality} and to solve $K$-SAT problems \cite{campos2021quantum}, and they are universal for quantum computation \cite{childs2009universal,lahini2018quantum}.

A number of works have reported equivalent results for Laplacian CTQWs when the fully connected vertex is involved. By fully connected vertex we mean a vertex which is adjacent (connected) to all the other vertices of the graph, as shown in Figure \ref{fig:graphs}. The dynamics of the central vertex of the star graph and that of any vertex of the complete graph are equivalent, showing periodic perfect revivals and strong localization on the initial vertex \cite{xu2009exact}, even in the presence of a perturbation $\lambda L^2$ \cite{candeloro2020continuous}. The spatial search of a marked vertex on the complete graph or on the star graph, when the target is the central vertex, are equivalent \cite{benedetti2019continuous}, and the same qualitative results are observed even in the presence of weak random telegraph noise \cite{cattaneo2018quantum}. The quantum-classical dynamical distance is a fidelity-based measure introduced to quantify the differences in the dynamics of classical versus quantum walks on a graph. Such distance turns out to be the same for the complete, star, and wheel graphs when the central vertex is assumed as the initial state for the walker \cite{benedetti2020quantum-classical}.

In this paper we prove the universality of the fully connected vertex in Laplacian CTQWs. This means that when the fully connected vertex of a graph is the initial state of the walk, or when it is the marked vertex (target) of a Grover-like CTQWs (those involved in spatial search or quantum transport), results do not depend on the considered graph $G$. In other words, those problems formulated on $G$ of order $N$ and on the complete graph of the same order, $K_N$, are equivalent. The present work thus explains the equivalent results between star, wheel, and complete graphs already observed and reported in literature, generalizing the equivalence to the fully connected vertices of any simple graph.

The paper is organized as follows. In Section \ref{sec:ctqw} we recall the CTQW model. In Section \ref{sec:dimredmet} we 
briefly review the dimensionality reduction method for quantum walks \cite{novo2015systematic}, according to which in Section \ref{sec:equiv_CTQW} we prove the equivalence of the Laplacian CTQW of a walker initially localized at a fully connected vertex in any simple graph. Instead, the corresponding CTQWs generated by the adjacency matrix do depend on the graph chosen. Then, in Section \ref{sec:equiv_CTQW_pbm_single} we prove that the equivalence applies also to Grover-like CTQWs for a single fully connected marked vertex, focusing on spatial search and quantum transport. In Section \ref{sec:equiv_CTQW_pbm_multiple} we generalize the result to the case of multiple marked vertices. 
Finally, we present our concluding remarks in Section \ref{sec:conclusion}.

\begin{figure}[!t]
\begin{indented}\item[]
\includegraphics[width=0.5\textwidth]{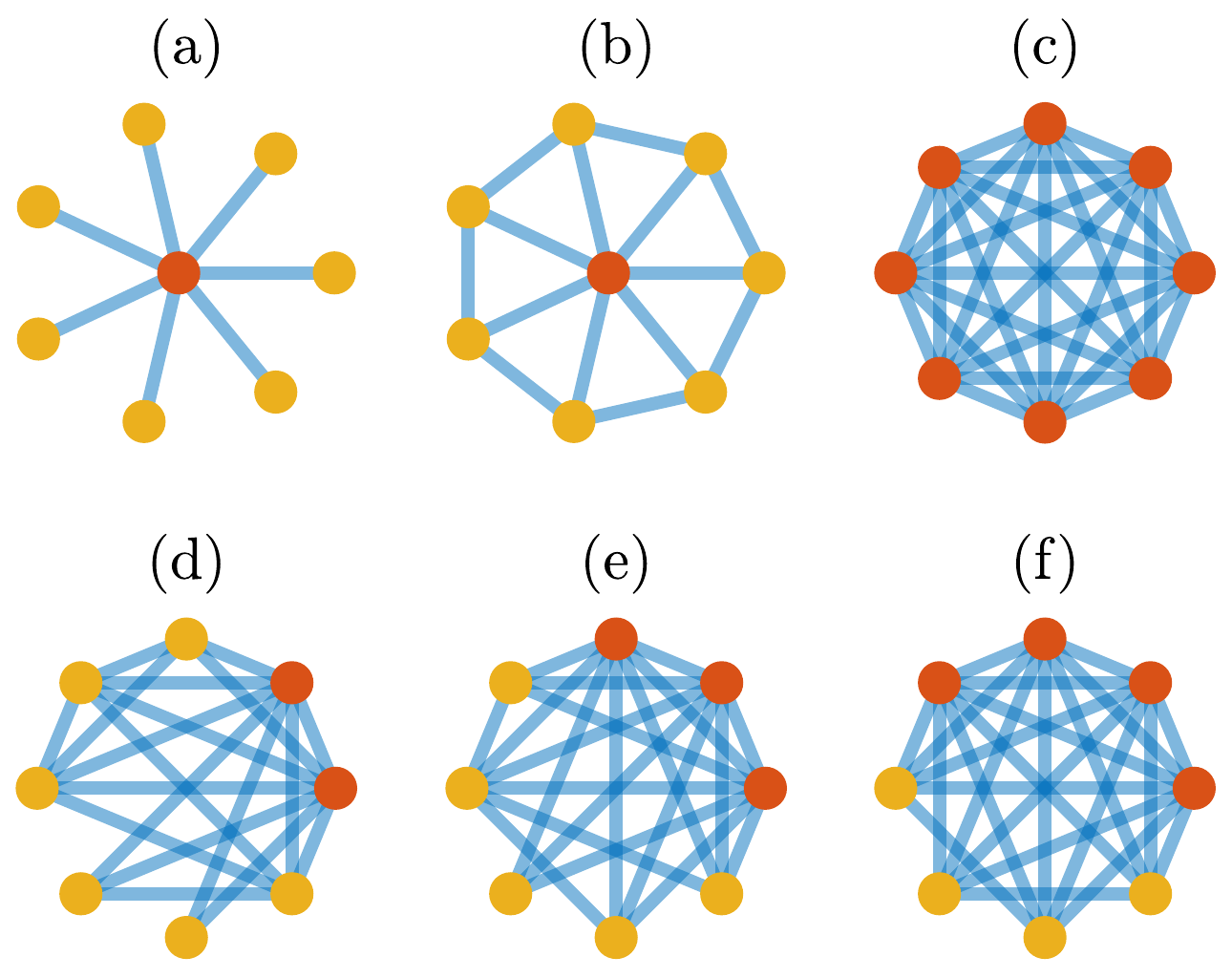}
\caption{Examples of graphs of order $N=8$ with at least one fully connected vertex $w$ (orange colored), $\mathrm{deg}(w)=N-1$. (a) Star graph $S_N$, (b) Wheel graph $W_N$, and (c) Complete graph $K_N$. (d)-(f) Random graphs.}
\label{fig:graphs}
\end{indented}
\end{figure}

\section{Continuous-time quantum walks}
\label{sec:ctqw}
A graph is a pair $G=(V,E)$, where $V$ denotes the non-empty set of vertices and $E$ the set of edges. The order of the graph is the number of vertices, $\vert V \vert =N$. A simple graph is an undirected graph containing no self loops or multiple edges. It is mathematically represented by the Laplacian matrix $L=D-A$, where the adjacency matrix $A$ ($A_{vv'}=1$ if the vertices $v$ and $v'$ are connected, $0$ otherwise) is symmetric and describes the connectivity of $G$ and $D$ is the diagonal degree matrix with $D_{vv}=\mathrm{deg}(v)=:d_v$ the degree of vertex $v$. According to this, $L$ is real, symmetric, positive semidefinite, and singular ($L$ always admits the null eigenvalue because every row sum and column sum of $L$ is zero, thus $\det(L)=0$).\footnote{There are a number of different, all related, definitions of Laplacian of a graph. Sometimes it is useful to normalize the Laplacian matrix $L$ to mitigate the weight of highly connected vertices. Indeed, a large degree results in large diagonal entry, $L_{vv}=d_v$, which dominates the matrix properties because much larger than the off-diagonal entries, $L_{vv'}=0,1$. The two matrices commonly known as normalized graph Laplacians are defined as $\mathcal{L}_{rw}:= D^{-1}L$ (closely related to a random walk) and $\mathcal{L}_{sym}:=D^{-1/2} L D^{-1/2}$ (symmetric matrix), with the convention that $D^{-1}_{vv}=0$ for $d_v = 0$ (i.e., $v$ is an isolated vertex) \cite{chung1997spectral,vonLuxburg2007tutorial}.}

The CTQW is the propagation of a free quantum particle when confined to a discrete space, e.g., a graph. The CTQW on a graph $G$ takes place on a $N$-dimensional Hilbert space $\mathcal{H}=\mathrm{span}(\{\vert v \rangle \mid v \in V\})$, and the kinetic energy term $-\nabla^2/2m$ is replaced by $\gamma L$, where $\hbar=1$ and $\gamma \in \mathbb{R}^+$ is the hopping amplitude of the walk. The state of the walker obeys the Schr\"{o}dinger equation
\begin{equation}
i\frac{d}{dt} \vert \psi(t) \rangle =H\vert \psi(t) \rangle
\label{eq:schrodinger}
\end{equation}
with Hamiltonian $H=\gamma L$. Hence, a walker starting in the state $\vert \psi_0 \rangle \in \mathcal{H}$ continuously evolves in time according to
\begin{equation}
\vert \psi(t) \rangle =U(t)\vert \psi_0 \rangle \,,
\label{eq:psit}
\end{equation}
with $U(t)=\exp[-i H t]$ the unitary time-evolution operator. The probability to find the walker in a target vertex $w$ is therefore $\vert \langle w \vert \exp\left[-i H t\right]\vert \psi_0 \rangle \vert^2$.

\section{Dimensionality reduction method}
\label{sec:dimredmet}
\subsection{Method}
In most CTQW problems encoded on a graph $G$ and a Hamiltonian $H$, the quantity of interest is the probability amplitude at a certain vertex of $G$. The graph often contains symmetries that allow us to simplify the problem, reducing the effective dimensionality of the latter. Indeed, the evolution of the system relevant to the problem actually occurs in a subspace, also known as Krylov subspace \cite{jafarizadeh2007investigation}, of the complete $N$-dimensional Hilbert space $\mathcal{H}$ spanned by the vertices of $G$.  This subspace contains the vertex of interest and it is invariant under the unitary time evolution. As a result, the original graph encoding the problem can be mapped onto an equivalent weighted graph of lower order, whose vertices are the basis states of the invariant subspace. The reduced Hamiltonian, i.e., $H$ written in the basis of the invariant subspace, still fully describes the dynamics relevant to the given problem. We can determine the invariant subspace and its basis by means of the dimensionality reduction method for CTQW \cite{novo2015systematic}, which we briefly review.

The unitary evolution  \eqref{eq:psit} can be expressed as
\begin{equation}
\vert \psi(t) \rangle = \sum_{n=0}^\infty \frac{(-it)^n}{n!}H^n\vert \psi_0 \rangle\,,
\end{equation}
so $\vert \psi(t) \rangle$ is contained in the subspace $\mathcal{I}(H,\vert \psi_0 \rangle) = \mathrm{span}(\lbrace H^n \vert \psi_0 \rangle \mid n \in \mathbb{N}_0\rbrace)$. This subspace of $\mathcal{H}$ is invariant under the action of the Hamiltonian and, thus, also of the unitary evolution. Naturally, $\dim \mathcal{I}(H,\vert \psi_0 \rangle )\leq \dim \mathcal{H}=N$. If the Hamiltonian is highly symmetrical, then only a small number of powers of $H^n\vert \psi_0 \rangle$ are linearly independent, hence the dimension of $\mathcal{I}(H,\vert \psi_0 \rangle)$ can be much smaller than $N$.

Let $P$ be the projector onto $\mathcal{I}(H,\vert \psi_0 \rangle)$. Then
\begin{equation}
U(t) \vert \psi_0 \rangle = e^{-iH_\mathrm{red}t}\vert \psi_0 \rangle\,,
\label{eq:reduced_time_evol}
\end{equation}
where $H_\mathrm{red}=PHP$ is the reduced Hamiltonian. We obtain this using the power series of $U(t)$ and the fact that $P^2=P$ (projector), $P\vert  \psi_0 \rangle =\vert \psi_0 \rangle$, and $P U(t) \vert \psi_0 \rangle = U(t) \vert \psi_0 \rangle$.

For any state $\vert \phi \rangle \in \mathcal{H}$, solution of the CTQW problem, we have
\begin{equation}
\langle \phi \vert U(t) \vert \psi_0 \rangle = \langle \phi_\mathrm{red} \vert e^{-iH_\mathrm{red}t}\vert \psi_0 \rangle \,,
\end{equation}
where $\vert \phi_\mathrm{red} \rangle = P\vert \phi \rangle$ is the reduced state. Analogously, using the projector $P'$ onto the subspace $\mathcal{I}(H,\vert \phi \rangle)$, we obtain
\begin{equation}
\langle \phi \vert U(t) \vert \psi_0 \rangle = \langle \phi \vert e^{-iH_\mathrm{red}'t}\vert {\psi_0}_\mathrm{red} \rangle \,,
\label{eq:prob_ampli_target}
\end{equation}
with $H_\mathrm{red}'=P'HP'$ and $\vert{\psi_0}_\mathrm{red} \rangle=P'\vert \psi_0 \rangle$.

An orthonormal basis of $\mathcal{I}(H,\vert \phi \rangle)$, say $\{\vert e_1 \rangle ,\ldots,\vert e_m \rangle \}$, can be iteratively obtained, as follows: $\vert e_1 \rangle := \vert \phi \rangle$, then $\vert e_{n+1} \rangle$ follows from orthonormalizing $H\vert e_n \rangle$ with respect to the previously obtained basis states, $\{\vert e_k \rangle \}_{k=1,\ldots,n}$, i.e.,
\begin{equation}
\vert u_{n+1} \rangle := H\vert e_n \rangle - \sum_{k=1}^n \langle e_k \vert H\vert e_n \rangle \vert e_k \rangle \quad \Rightarrow \quad \vert e_{n+1} \rangle := \frac{\vert u_{n+1} \rangle} {\Vert \vert u_{n+1} \rangle \Vert}\,.
\end{equation}
The procedure stops when we find the minimum $m$ such that $H\vert e_m \rangle \in \mathrm{span}(\{\vert e_1 \rangle,\ldots,\vert e_m \rangle\})$. The projector onto $\mathcal{I}(H,\vert \phi \rangle)$ is therefore $P'=\sum_{n=1}^m \vert e_n \rangle \langle e_n \vert$.

\subsection{Complete Graph}
\label{subsec:complete_graph}
As an example, we review the well-known reduced problem of the CTQW on the complete graph on $N$ vertices, $K_N$, when generated by the Laplacian matrix or by the adjacency matrix. Each pair of vertices is connected by an edge, so any vertex is fully connected and has degree $N-1$. The adjacency matrix is $(A_{K})_{vv'} = 1$ $\forall v \neq v'$, the diagonal degree matrix is $D_K = (N-1)I$, where $I$ is the identity operator, and the Laplacian matrix is $L_K = D_K-A_K$. Suppose we want to study the CTQW of a walker initially localized at a certain vertex $w$ or, alternatively, for walker starting from any other initial state, to compute the probability amplitude at $w$. The invariant subspace relevant to problem is
\begin{eqnarray}
\mathcal{I}(L_K,\vert w \rangle) &= \mathcal{I}(A_K,\vert w \rangle)\nonumber\\
&= \mathrm{span}\left(\left\lbrace \vert e_1 \rangle = \vert w \rangle ,\vert e_2 \rangle = \frac{1}{\sqrt{N-1}}\sum\nolimits_{v\neq w} \vert v \rangle\right\rbrace\right)\,.
\label{eq:basis_inv_subsp_K}
\end{eqnarray}
Writing $L_K$ and $A_K$ in this subspace, we find, respectively, the reduced Laplacian matrix
\begin{equation}
L_{K,\mathrm{red}} =
\left(\begin{array}{cc}
N-1 & -\sqrt{N-1}\\
-\sqrt{N-1} & 1
\end{array}\right)\,,
\label{eq:LK_red_gen}
\end{equation}
and the reduced adjacency matrix \cite{novo2015systematic}
\begin{equation}
A_{K,\mathrm{red}} =
\left(\begin{array}{cc}
0 & \sqrt{N-1}\\
\sqrt{N-1} & N-2
\end{array}\right)\,.
\label{eq:AK_red_gen}
\end{equation}
It is worth noticing that, consistently with $L_K=D_K-A_K$, we have $L_{K,\mathrm{red}}=D_{K,\mathrm{red}}-A_{K,\mathrm{red}}$, since $D_K$ written in the basis \eqref{eq:basis_inv_subsp_K} is $D_{K,\mathrm{red}}=(N-1)I_{2\times2}$.

The steps required to obtain the orthonormal basis \eqref{eq:basis_inv_subsp_K}, the reduced Laplacian matrix \eqref{eq:LK_red_gen}, and the reduced adjacency matrix \eqref{eq:AK_red_gen} for the complete graph are the same as those presented, in a more general case, in the proofs of Theorem \ref{th:Lapl} and Proposition \ref{prop:Adj}, to which we refer the reader for details.

\section{Universality of a CTQW starting from a fully connected vertex}
\label{sec:equiv_CTQW}
In this section we discuss the CTQW generated either by the Laplacian matrix, $H=\gamma L$, or by the adjacency matrix, $H=\gamma A$. The hopping amplitude $\gamma$ plays the role of a time scaling factor in the time-evolution operator $\exp[-i L \gamma t]$ or $\exp[-i A \gamma t]$. Therefore, in the following we set $\gamma = 1$ so that, together with $\hbar = 1$, time and energy are dimensionless.

\subsection{Laplacian CTQW}
We will refer to the CTQW generated by the Laplacian matrix $L$ as a Laplacian CTQW.
\begin{theorem}
\label{th:Lapl}
Let $G=(V,E)$ be a simple graph on $N=\vert V \vert$ vertices and $M=\vert E \vert$ edges, with Laplacian matrix $L_G=D-A$. Let $w\in V$ be a fully connected vertex of $G$, with degree $d_w=N-1$. Then, the time-evolution of $\vert w \rangle$ under the Laplacian matrix is
\begin{equation}
    e^{-i L_G t}\vert w \rangle = e^{-i L_{G,\mathrm{red}} t}\vert w \rangle \,,
    \label{eq:L_thesis_evol}
\end{equation}
is entirely contained in the invariant subspace
\begin{equation}
\mathcal{I}(L_G,\vert w \rangle) = \mathrm{span}\left(\left\lbrace \vert e_1 \rangle = \vert w \rangle ,\vert e_2 \rangle = \frac{1}{\sqrt{N-1}}\sum\nolimits_{v\neq w} \vert v \rangle\right\rbrace\right)\,,
\label{eq:basis_inv_subsp_G}
\end{equation}
and is generated by the reduced Laplacian matrix
\begin{equation}
L_{G,\mathrm{red}} =
\left(\begin{array}{cc}
N-1 & -\sqrt{N-1}\\
-\sqrt{N-1} & 1
\end{array}\right)\,.
\label{eq:LG_red_gen}
\end{equation}
\end{theorem}

\begin{remark}
We emphasize that $\dim \mathcal{I} (L_G,\vert w \rangle) = 2 \leq \dim \mathcal{H}=N$ independently of $N$ and of the graph considered. Theorem \ref{th:Lapl} generalizes what already known for the complete graph in Section \ref{subsec:complete_graph}, proving that the CTQW of the fully connected vertex $\vert w \rangle$ is independent of the graph.
\end{remark}

\begin{proof}
Let $\mathcal{H}$ be the $N$-dimensional Hilbert space of a quantum walker on $G$. The time evolution of the state $\vert w \rangle $ generated by $L_G$, $\exp\left[-i L_G t\right] \vert w \rangle $, 
belongs to a subspace of $\mathcal{H}$,
\begin{equation}
\mathcal{I}(L_G,\vert w \rangle ) := \mathrm{span}(\lbrace L_G^n \vert w \rangle  \mid n \in \mathbb{N}_0\rbrace)\,.
\end{equation}
The proof makes use of the dimensionality reduction method (Section \ref{sec:dimredmet}) and consists of two parts. (i) First, we prove Equation \eqref{eq:basis_inv_subsp_G}.
(ii) Second, we prove Equation \eqref{eq:LG_red_gen}. Therefore, if the CTQW of a fully connected vertex $w$ on any graph $G$ satisfy these two conditions, then the statement \eqref{eq:L_thesis_evol} follows from Equation \eqref{eq:reduced_time_evol}.

(i) The first basis state is $\vert e_1 \rangle =\vert w \rangle $. Then we consider
\begin{equation}
L_G \vert e_1 \rangle  = (N-1)\vert w \rangle -\sum_{v \neq w} \vert v \rangle =:  (N-1)\vert e_1 \rangle -\sqrt{N-1}\vert e_2 \rangle \,,
\label{eq:Th1_LGe1}
\end{equation}
where we have used the fact that $w$ is adjacent to all the other vertices, $d_w=N-1$. The basis state $\vert e_2 \rangle$ follows from orthonormalizing $L_G \vert e_1 \rangle$ with respect to the previous basis state, $\vert e_1 \rangle$.

To find the next basis state, we compute $L_G \vert e_2 \rangle $ and then we orthonormalize it with respect to the previous basis states. To compute the projections $\langle e_n \vert L_G \vert e_2 \rangle $, with $n=1,2$, it is convenient to use the definition of Laplacian matrix. From Equation \eqref{eq:Th1_LGe1} we have that
\begin{equation}
\langle e_1 \vert L_G\vert e_2 \rangle  = -\sqrt{N-1}\,,
\end{equation}
and 
\begin{eqnarray}
\langle e_2 \vert L_G\vert e_2 \rangle  &= \frac{1}{N-1}\sum_{v,v' \neq w} (D_{vv'}-A_{vv'})
= \frac{1}{N-1}\left[ \sum_{v\neq w} d_{v}-(2M-2d_w)\right]\nonumber\\
&= \frac{1}{N-1}\left[ 2M-(2M-d_w)\right]=1\,,
\end{eqnarray}
because $D$ is diagonal, $D_{vv'}=0$ for $v \neq v'$, and 
\begin{eqnarray}
\sum_{v,v' \neq w} A_{vv'}
&= \sum_{v\in V}\sum_{v' \neq w}A_{vv'}-\sum_{v' \neq w}A_{wv'}
= \sum_{v,v'\in V}A_{vv'}-\sum_{v\in V}A_{vw}-d_w\nonumber\\
&= 2M-2d_w\,,
\label{eq:sum_adj_not_w}
\end{eqnarray}
since $\sum_{v' \neq v}A_{vv'} = \sum_{v' \in V} A_{vv'} = d_v$ in a graph with no self loops (a vertex is not adjacent to itself), as in the present case. Summing all the elements of the adjacency matrix, as well as summing the degrees, means counting the edges twice, $\sum_{v,v' \in V}A_{vv'}=\sum_{v\in V} d_{v}=2M$ with $M$ the number of edges. In graph theory the latter is known as the degree sum formula and it implies the \textit{handshaking lemma}. We can now prove that 
\begin{equation}
L_G \vert e_2 \rangle  =- \sqrt{N-1}\vert e_1 \rangle +\vert e_2 \rangle \,,
\label{eq:LGe1_lincomb}
\end{equation}
therefore that $L_G^n \vert w \rangle  \in \mathrm{span}(\lbrace \vert e_1 \rangle ,\vert e_2 \rangle \rbrace)\forall n \in \mathbb{N}_0$, by showing that 
\begin{equation}
\vert \lambda \rangle :=(L_G-I) \vert e_2 \rangle  + \sqrt{N-1}\vert e_1 \rangle = 0\,,
\label{eq:Le1_ortho}
\end{equation}
where $I$ is the identity. First, we project it onto $\vert w \rangle $
\begin{eqnarray}
\langle w \vert \lambda \rangle
&= \frac{1}{\sqrt{N-1}} \left( \sum_{v \neq w}(d_v-1)\delta_{wv}-\sum_{v \neq w}A_{wv}+N-1\right)\nonumber\\
&= \frac{1}{\sqrt{N-1}} \left( 0-d_w+N-1\right) = 0\,,
\end{eqnarray}
and then we project it onto any other vertex state, $\vert v' \neq w \rangle$,
\begin{eqnarray}
\langle v' \vert \lambda \rangle
&= \frac{1}{\sqrt{N-1}} \left[ \sum_{v \neq w}(d_v-1)\delta_{v'v}-\sum_{v \neq w}A_{v'v}+0\right]\nonumber\\
&= \frac{1}{\sqrt{N-1}} \left[ d_{v'}-1- \left(\sum_{v\in V}A_{v'v}-A_{v'w} \right)\right]\nonumber\\
&= \frac{1}{\sqrt{N-1}} \left[ d_{v'}-1- \left(d_{v'}-1 \right)\right] = 0\,,
\end{eqnarray}
where $A_{v'w}=1$ because $w$ is adjacent to all the other vertices. This proves Equation \eqref{eq:Le1_ortho}, because the $w$-th component and any other component, $v'\neq w$, are null.  The statement \eqref{eq:basis_inv_subsp_G} follows.

(ii) We can easily prove Equation \eqref{eq:LG_red_gen} by taking the matrix elements
\begin{equation}
(L_{G,\mathrm{red}})_{jk}:=\langle e_j \vert L_G \vert e_k \rangle =(L_{G,\mathrm{red}})_{kj}\,,
\end{equation}
with $j,k=1,2$, from (i).

To summarize, the time evolution of the fully connected vertex state $\vert w \rangle$ always belongs to the subspace \eqref{eq:basis_inv_subsp_G} and is fully described by the reduced generator \eqref{eq:LG_red_gen} indipendently of the graph $G$ considered. This proves the statement \eqref{eq:L_thesis_evol}, concluding the proof.
\end{proof}

\begin{corollary}
Let us consider the Laplacian CTQWs on a graph $G_1$ and on a graph $G_2$ both of order $N$ with a fully connected vertex $w$. Let us assume that the initial states are $\vert \psi_{0,G_1} \rangle$ and $\vert \psi_{0,G_2} \rangle$, respectively. Then, the probability amplitude of finding the walker at $w$ is the same, 
$\langle w \vert \exp\left[-i L_{G_1} t\right]\vert \psi_{0,G_1} \rangle  = \langle w \vert \exp\left[-i L_{G_2} t\right]\vert \psi_{0,G_2} \rangle $, provided that the two initial states have the same projection onto the subspace $\mathcal{I}(L_{G_1}, \vert w \rangle)$ \eqref{eq:basis_inv_subsp_G}.
\end{corollary}
\begin{proof}
This directly follows from Equation \eqref{eq:prob_ampli_target}, with $\vert \phi \rangle = \vert w \rangle $, and Theorem \ref{th:Lapl}.
\end{proof}

\subsection{Adjacency CTQW}
We will refer to the CTQW generated by the adjacency matrix $A$ as an adjacency CTQW.
\begin{proposition}
\label{prop:Adj}
Let $G=(V,E)$ be a simple graph on $N=\vert V \vert$ vertices and $M=\vert E \vert$ edges, with adjacency matrix $A_G$. Let $w\in V$ be a fully connected vertex of $G$, with degree  $d_w=N-1$. Then, the adjacency CTQW of the state $\vert w \rangle $ does depend on the graph $G$ considered.
\end{proposition}

\begin{proof}
The proof makes use of the dimensionality reduction method (Section \ref{sec:dimredmet}) and consists of three parts. (i.a) First, we prove that
\begin{equation}
\dim \mathcal{I}(A_G,\vert w \rangle ) \geq 2 = \dim \mathcal{I}(A_K,\vert w \rangle )\,,
\label{eq:A_thesis2}
\end{equation}
where the subscript $K$ refers to the complete graph and, as known, $\mathcal{I}(A_K,\vert w \rangle )$ is \eqref{eq:basis_inv_subsp_K}. This is a first indication that the CTQW of $\vert w \rangle $ generated by $A_G$ and $A_K$ are not equivalent, in general, revealing a first dependence on the graph considered. (i.b) In particular, if the graph $G$ has more than one fully connected vertex and $G \neq K_N$, then $\dim \mathcal{I}(A_G,\vert w \rangle ) > 2$. (ii) Second, we prove that even if $\mathcal{I}(A_G,\vert w \rangle )= \mathcal{I}(A_K,\vert w \rangle )$, the two reduced generators are different, $A_{G,\mathrm{red}}\neq A_{K,\mathrm{red}}$, and thus lead to different time evolutions.

(i.a) The first basis state is $\vert e_1 \rangle =\vert w \rangle $. Then we consider
\begin{equation}
A_G \vert e_1 \rangle  = \sum_{v \neq w} \vert v \rangle =: \sqrt{N-1} \vert e_2 \rangle\,,
\label{eq:Th1_AGe1}
\end{equation}
and $\vert e_2 \rangle$ follows from normalizing $A_G \vert e_1 \rangle$, as the latter is already orthogonal to $ \vert e_1 \rangle$.

To find the next basis state, we compute $A_G \vert e_2 \rangle $ and then we orthonormalize it with respect to the previous basis states. To compute the projections $\langle e_n \vert A_G \vert e_2 \rangle $, with $n=1,2$, it is convenient to use the definition of adjacency matrix. From Equation \eqref{eq:Th1_AGe1} we have that
\begin{equation}
\langle e_1 \vert A_G\vert e_2 \rangle = \sqrt{N-1}\,,
\label{eq:e1AGe2}
\end{equation}
and, using Equation \eqref{eq:sum_adj_not_w},
\begin{equation}
\langle e_2 \vert A_G\vert e_2 \rangle = \frac{1}{N-1}\sum_{v,v' \neq w} A_{vv'}= \frac{1}{N-1}\left(2M-2d_w\right)=\frac{2M}{N-1} -2\,,
\end{equation}
where, we recall, $M$ is the number of edges. We can now study whether or not the state
\begin{equation}
\vert \alpha \rangle := \left[A_G-\left(\frac{2M}{N-1} -2\right)\right] \vert e_2 \rangle  -\sqrt{N-1}\vert e_1 \rangle \,
\label{eq:Ae1_alpha}
\end{equation}
is null. If null, then the invariant subspace has dimension 2, as $A_G \vert e_2 \rangle$ is a linear combination of $\vert e_1 \rangle$ and $\vert e_2 \rangle$, otherwise it has dimension $>2$. First, we project the state \eqref{eq:Ae1_alpha} onto $\vert w \rangle $, observing that $\langle w \vert \alpha \rangle = 0$ from Equation \eqref{eq:e1AGe2}, and then we project it onto any other vertex state, $\vert v' \neq w \rangle$,
\begin{eqnarray}
\langle v' \vert \alpha \rangle
&= \frac{1}{\sqrt{N-1}} \left[ \sum_{v \neq w}A_{v'v}-\left(\frac{2M}{N-1} -2\right)\sum_{v \neq w}\delta_{v'v}\right]\nonumber\\
&= \frac{1}{\sqrt{N-1}} \left[ (d_{v'}-1)-\frac{2M}{N-1} +2\right]\nonumber\\
&= \frac{1}{\sqrt{N-1}} \left( d_{v'}+1-\frac{2M}{N-1} \right)\,,
\label{eq:vAe1_on}
\end{eqnarray}
where $\sum_{v\neq w}A_{v'v}=\sum_{v\in V}A_{v'v}-A_{v'w}=d_{v'}-A_{v'w}$ and $A_{v'w}=1$ because $w$ is adjacent to all the other vertices. We have proved that the $w$-th component is null, but the other components $v'\neq w$ depend on $v'$, so they are not null, in general. According to this, $A_G\vert e_2 \rangle $ is not just a linear combination of $\vert e_1 \rangle $ and $\vert e_2 \rangle $, further basis states are required, and so the statement \eqref{eq:A_thesis2} follows.

(i.b) Let us now assume that there is another fully connected vertex $w' \neq w$, $d_{w'}=N-1$. Then, $\langle w \vert \alpha \rangle = 0$ still holds and Equation \eqref{eq:vAe1_on} for $v' = w'$ reads as
\begin{equation}
\langle w' \vert \alpha \rangle= \frac{1}{(N-1)^{3/2}} \left( N^2-N-2M \right)\,,
\label{eq:wAe1_on_2fully}
\end{equation}
which is null for $N = (1 \pm \sqrt{1+8M})/2$. However, $N\in \mathbb{N}$ requires the solution with the plus sign and $\sqrt{1+8M}=2m+1$, with $m\in\mathbb{N}_0$. Solving the latter condition with respect to $m$ leads to $m = [-1 \pm (2m+1)]/2$. The only acceptable solution is $m=m$, i.e., any positive odd number $2m+1$ can be written as $\sqrt{1+8M}$. The degree sum formula, $\sum_{v \in V} d_v = 2M$, allows us to write
\begin{equation}
N = \frac{1}{2}\left(1 + \sqrt{1+4\sum\nolimits_{v \in V}d_v}\right)\,.
\label{eq:N_ib_general}
\end{equation}
Now we study whether the Equation \eqref{eq:N_ib_general} admits a solution. The presence of fully connected vertex make the graph connected, and $d_v \geq 2$ $\forall v \in V$ since, by assumption, there are at least two fully connected vertices. The graph satisfying the minimal conditions is the graph with two fully connected vertices, $w,w'$ with $d_w=d_{w'}=N-1$, and with all the other $N-2$ vertices connected only to $w$ and $w'$, $d_v =2$ $\forall v \neq w,w'$. Hence, $\sum_{v \in V}d_v = 2(N-1)+(N-2)2$, from which the right-hand side of Equation \eqref{eq:N_ib_general} is
\begin{equation}
f(N) = \frac{1}{2}\left(1 + \sqrt{16N-23}\right)\,.
\label{eq:N_ib_f}
\end{equation}
If we assume that all the vertices are fully connected, then we get the complete graph. Hence, $\sum_{v \in V}d_v = N(N-1)$, from which Equation \eqref{eq:N_ib_general} holds for any $N$. However, we are interested in graphs other than the complete one. There is no graph with only $N-1$ fully connected vertices, as, otherwise, the remaining vertex is necessarily connected to all the others and so the graph is complete. There is, however, the graph with $N-2$ fully connected vertices, obtained by removing one edge from the complete graph. The two non-fully connected vertices thus obtained have degree $N-2$.  Hence, $\sum_{v \in V}d_v = 2(N-2)+(N-2)(N-1)$, from which the right-hand side of Equation \eqref{eq:N_ib_general} is
\begin{equation}
g(N) =  \frac{1}{2}\left(1 + \sqrt{4N^2-4N-7}\right)\,.
\label{eq:N_ib_g}
\end{equation}
All the possible graphs on $N$ vertices having a number $2 \leq \mu \leq N-2$ of fully connected vertices fall within these two cases. In Figure \ref{fig:sol_eq_adj_proof} we study Equation \eqref{eq:N_ib_general}, and we observe that there are no solutions, as none of the right-hand sides, $f(N)$ and $g(N)$, have intersection with the left-hand side, the line $h(N)=N$. We have just proved that, under the assumption of having at least two fully connected vertices and $G \neq K_N$, more than two basis states are required, therefore $\dim \mathcal{I}(A_G,\vert w \rangle ) >2$. Indeed, while the $w$-th component is null, the components corresponding to the other fully connected vertex (or vertices) $w'\neq w$ \eqref{eq:wAe1_on_2fully} are not, thus $A_G\vert e_2 \rangle $ is not just a linear combination of $\vert e_1 \rangle $ and $\vert e_2 \rangle $.

\begin{figure}[!t]
\begin{indented}\item[]
\includegraphics[width=0.5\textwidth]{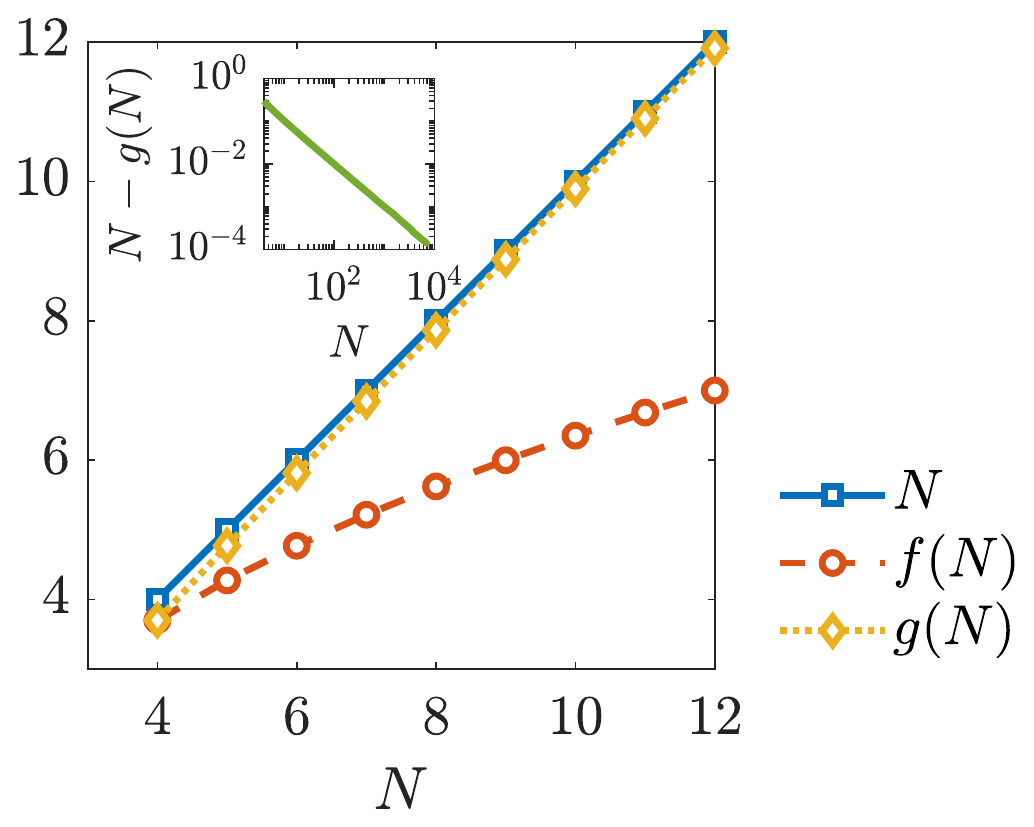}
\caption{Graphical solution of Equation \eqref{eq:N_ib_general}. The left-hand side (LHS) is $N$ (blue solid line, square). The right-hand side (RHS) is $f(N)$ \eqref{eq:N_ib_f} (orange dashed line, circles) or $g(N)$ \eqref{eq:N_ib_g} (yellow dotted line, diamonds). All the possible graphs on $N$ vertices having a number $2 \leq \mu \leq N-2$ of fully connected vertices result in a RHS which falls within these two cases. Results are shown for $N\geq 4$, because the graphs for $N=2,3$ and $\mu =2$ would be the complete graph $K_2,K_3$, respectively. We observe that there are no intersections between the RHS and the LHS, as highlighted in the log-log plot of $N-g(N)$ in the inset. Note that, $g(N)\sim N$ for large $N$, but $g(N)$ never reaches $N$. Therefore, Equation \eqref{eq:N_ib_general} has no solution.}
\label{fig:sol_eq_adj_proof}
\end{indented}
\end{figure}

(ii) Let us now assume that there is only one fully connected vertex, $w$. Then, $\langle w \vert \alpha \rangle = 0$ still holds and Equation \eqref{eq:vAe1_on} reads as
\begin{equation}
\langle v' \vert \alpha \rangle = \frac{1}{(N-1)^{3/2}} \left[(N-1)d_{v'}- \sum_{v \neq w} d_{v} \right]\,,
\end{equation}
since $\sum_{v \in V}d_v = \sum_{v \neq w}d_v +(N-1)=2M$. The above expression is null if $(N-1)d_{v'}= \sum_{v \neq w} d_{v}$ and the latter condition must apply $\forall v'\neq w$ to make the state \eqref{eq:Ae1_alpha} null. Therefore, this condition implies that all the vertices, except $w$, must have the same degree $d_{v'}$. This is the case, e.g., of the star graph [Figure \ref{fig:graphs}(a)] or the wheel graph [Figure \ref{fig:graphs}(b)].
We have just proved that if a simple graph has one fully connected vertex, $w$, and $\mathrm{deg}(v)=d$ $\forall v \in V\setminus \lbrace w \rbrace$, then all the components of the state \eqref{eq:Ae1_alpha} are null. Hence, $A_G^n \vert w \rangle  \in \mathrm{span}(\lbrace \vert e_1 \rangle ,\vert e_2 \rangle \rbrace)\forall n \in \mathbb{N}_0$, because
\begin{equation}
A_G \vert e_2 \rangle  = \sqrt{N-1}\vert e_1 \rangle +\left(\frac{2M}{N-1} -2\right)\vert e_2 \rangle \,,
\end{equation}
and therefore $\mathcal{I}(A_G,\vert w \rangle )=\mathcal{I}(A_K,\vert w \rangle )$. 
From (i) we have the matrix elements
\begin{equation}
(A_{G,\mathrm{red}})_{jk}:=\langle e_j \vert A_G \vert  e_k \rangle = (A_{G,\mathrm{red}})_{kj}\,,
\end{equation}
with $j,k=1,2$. Writing $A_G$ in the basis $\lbrace \vert e_1 \rangle ,\vert e_2 \rangle \rbrace$, we find that
\begin{equation}
A_{G,\mathrm{red}} =
\left( \begin{array}{cc}
0 & \sqrt{N-1}\\
\sqrt{N-1} & \frac{2M}{N-1}-2
\end{array}\right)\,.
\label{eq:AG_red_gen}
\end{equation} 
The reduced generator $A_{G,\mathrm{red}}$  \eqref{eq:AG_red_gen} differs from $A_{K,\mathrm{red}}$  \eqref{eq:AK_red_gen} in the element $(A_{\mathrm{red}})_{22}$. We observe that 
\begin{equation}
\frac{2M}{N-1}-2 = N-2 \Leftrightarrow M = \frac{N(N-1)}{2}\,,
\end{equation}
but only the complete graph has $M = N(N-1)/2$ edges. Moreover, also the star graph and the wheel graph differ in that element, as $M=N-1$ and $M=2(N-1)$, respectively. So, the adjacency CTQW on the graphs which are regular except for the fully connected vertex $w$ are neither equivalent among them, in general, nor to the adjacency CTQW on the complete graph. The reason is that the reduced generators, $A_{G,\mathrm{red}}$ and $A_{K,\mathrm{red}}$, are different, as they depend on the number of edges $M$, and thus they lead to different time evolutions, which, however, belong to the same invariant subspace $\mathcal{I}(A_G,\vert w \rangle )=\mathcal{I}(A_K,\vert w \rangle )$.

To summarize, adjacency CTQWs do depend on the given graph $G$. Considering the adjacency CTQWs of the fully connected vertex state $\vert w \rangle$ either the time evolutions of it belong to different subspaces (see Equation \eqref{eq:A_thesis2}) or, otherwise, the reduced generators are different, as they depend on the number of edges $M$. This proves the Proposition \ref{prop:Adj}, concluding the proof. 
\end{proof}

\section{Grover-like CTQWs with single marked vertex}
\label{sec:equiv_CTQW_pbm_single}
\begin{corollary}
\label{cor:Lctqw_pbm}
Let $w$ be a fully connected marked vertex of a simple graph $G$ of order $N$ with Laplacian matrix $L$. Let us consider the Grover-like CTQW where the quantity of interest is the probability amplitude at $w$. Let
\begin{equation}
H = \gamma L + \lambda \vert w \rangle\langle w \vert
\label{eq:Ham_pbm}
\end{equation}
be the Hamiltonian encoding the problem, where $\gamma \in \mathbb{R}^+$, $\lambda \in \mathbb{C}$, and $H_w := \lambda \vert w \rangle\langle w \vert $ is the oracle Hamiltonian. Then, given the initial state $\vert \psi_0 \rangle$, the probability amplitude at the marked vertex is $\langle w \vert \exp\left[ -i H_{\mathrm{red}} t\right] \vert {\psi_0}_\mathrm{red} \rangle$, where $\vert{\psi_0}_\mathrm{red} \rangle = P\vert \psi_0 \rangle$ with $P$ the projector onto the invariant subspace $\mathcal{I}(H,\vert w \rangle)$ \eqref{eq:basis_inv_subsp_G} relevant to the problem and the reduced Hamiltonian is
\begin{equation}
H_{\mathrm{red}} =\gamma
\left( \begin{array}{cc}
N-1+\lambda/\gamma & -\sqrt{N-1}\\
-\sqrt{N-1} & 1
\end{array} \right)\,.
\label{eq:Ham_pbm_red}
\end{equation}
Grover-like CTQWs on a graph $G_1$ and on a graph $G_2$ both of order $N$ result in the same probability amplitude $\langle w \vert \exp\left[ -i H_{G_1} t\right] \vert \psi_{0,G_1} \rangle=\langle w \vert \exp\left[ -i H_{G_2} t\right] \vert \psi_{0,G_2} \rangle$ provided that $\vert{\psi_{0,G_1}}_\mathrm{red} \rangle = \vert{\psi_{0,G_2}}_\mathrm{red} \rangle$.
\end{corollary}

\begin{proof}
First, we prove that the invariant subspace $\mathcal{I}(H,\vert w \rangle )$  relevant to the problem is \eqref{eq:basis_inv_subsp_G} and then that the reduced Hamiltonian is \eqref{eq:Ham_pbm_red}. The only effective parameter in the Hamiltonian \eqref{eq:Ham_pbm} is the ratio $\lambda/\gamma$. Writing $H = \gamma H'$ we understand that $\gamma$ only determines the timescale of the evolution. Clearly $\mathcal{I}(H,\vert w \rangle )=\mathcal{I}(H',\vert w \rangle )$ and $\vert e_1 \rangle =\vert w \rangle $. The oracle $H_w'=(\lambda/\gamma) \vert w \rangle\langle w \vert $ acts nontrivially only onto $\vert e_1 \rangle $. Therefore, after orthonormalizing $H'\vert e_1 \rangle $ with respect to $\vert e_1 \rangle $, we find the second basis state, $\vert e_2 \rangle $ defined in Equation \eqref{eq:basis_inv_subsp_G}. We observe that $H'\vert e_2 \rangle = L \vert e_2 \rangle$, as $H_w'\vert e_2 \rangle = 0$, thus, according to the proof of Theorem \ref{th:Lapl}, there are no further basis states. Hence, the dynamics relevant to the Grover-like CTQWs for the fully connected vertex belong to the subspace \eqref{eq:basis_inv_subsp_G}. The oracle Hamiltonian $H_w$ has a natural representation in such subspace
\begin{equation}
H_{w,\mathrm{red}} = \lambda \vert e_1 \rangle\langle e_1 \vert  =
\left( \begin{array}{cc}
\lambda & 0\\ 0 & 0
\end{array} \right)\,.
\label{eq:Hw_red}
\end{equation}
The reduced Hamiltonian \eqref{eq:Ham_pbm_red} follows from summing the reduced Laplacian matrix \eqref{eq:LG_red_gen} and the reduced oracle Hamiltonian \eqref{eq:Hw_red}. The remark on the equal probability amplitudes at $w$ depending on the initial state follows from Equation \eqref{eq:prob_ampli_target}, with $\vert \phi \rangle = \vert w \rangle$.
\end{proof}

Grover-like CTQWs of great interest formulated as in the Corollary \ref{cor:Lctqw_pbm} are spatial search \cite{childs2004spatial}, $\lambda = -1$, and quantum transport, $\gamma =1$ and $\lambda = -i \kappa$, with $\kappa \in \mathbb{R^+}$ and $i=\sqrt{-1}$ the imaginary unit \cite{rebentrost2009environment}. In the former, solving the problem amounts to making the walker reach the state $\vert w \rangle $ with the maximum probability starting from the equal superposition of all vertices. In the latter, the quantity of interest is often the transport efficiency, $\eta = 2 \kappa \int_0^{+\infty} \langle w \vert  \rho(t) \vert w \rangle \,dt$, the integrated probability of trapping at the vertex $w$, where $\rho(t)$ is the density matrix of the walker. The transport efficiency can also be read as the complement to $1$ of the probability of surviving within the graph, i.e., $\eta = 1-\Tr\left[ \lim_{t\to+\infty}\rho(t) \right]$ \cite{razzoli2021transport}. We point out that whenever $\mathrm{Im}(\lambda) \neq 0$ the Hamiltonian \eqref{eq:Ham_pbm} is a non-Hermitian effective Hamiltonian that leads to non-unitary dynamics. This is useful to phenomenologically model certain processes like, if $\mathrm{Im}(\lambda)<0$, the dissipative dynamics in quantum optics \cite{plenio1998quantum-jump} or the absorption of an excitation in light harvesting systems \cite{olaya2008efficiency,caruso2009highly}.

\subsection{Spatial search}
\label{subsec:spatial_search_single}
The Hamiltonian encoding the problem is
\begin{equation}
    H = \gamma L - \vert w \rangle \langle w \vert\,,
    \label{eq:Ham_search_1}
\end{equation}
where the marked vertex, target of the search, is the fully connected vertex $w$. Since we have no information about the marked vertex, the initial state is commonly chosen as the equal superposition of all vertices, $\vert \psi_0 \rangle=\sum_{v\in V} \vert v \rangle /\sqrt{N}$. The goal is to tune the hopping amplitude $\gamma$ to maximize the probability amplitude at the marked vertex after a period of time of evolution. The time evolution of $\vert \psi_0 \rangle$ is entirely contained in $\mathcal{I}(H,\vert w \rangle )$, as $\vert \psi_0 \rangle = (\vert e_1 \rangle +\sqrt{N-1}\vert e_2 \rangle )/\sqrt{N}$ and so $\vert {\psi_0}_{\mathrm{red}} \rangle = \vert {\psi_0} \rangle$. Hence, not only the success probability of finding $w$, but also the entire dynamics of the system $\exp\left[-i Ht\right] \vert \psi_0 \rangle$ is the same on any simple graph $G$. 
According to Corollary \ref{cor:Lctqw_pbm}, the results we have for the spatial search on the complete graph, a well-known problem \cite{farhi1998analog,childs2004spatial,benedetti2019continuous}, also apply to the search of $w$ on other graphs. Therefore, if $\gamma =1/N$ (optimal value), then the walker reaches $w$ with probability
\begin{equation}
P_w(t) = \vert \langle w \vert e^{-iHt}\vert \psi_0 \rangle\vert^2 
=\frac{1}{N} \cos^2\left(\frac{t}{\sqrt{N}}\right) +\sin^2\left(\frac{t}{\sqrt{N}}\right) 
\label{eq:prob_search_w}
\end{equation}
equal to one (certainty) at time $t^\ast = \pi \sqrt{N}/2$. 

\subsection{Quantum transport}
\label{subsec:quantum_transport_single}
The non-Hermitian effective Hamiltonian encoding the problem is
\begin{equation}
    H = L - i \kappa \vert w \rangle \langle w \vert\,,
    \label{eq:Ham_trans_1}
\end{equation}
where the trapping vertex is the fully connected vertex $w$ and the trapping rate $\kappa \in \mathbb{R}^+$ ($\lambda = -i\kappa$ in \eqref{eq:Ham_pbm}). We assume that the initial state is localized at a vertex different from $w$, $\vert \psi_0 \rangle=\vert v \neq w \rangle$. Under such assumptions, the transport efficiency of the complete graph is $\eta_{K} = 1/(N-1)$ \cite{caruso2009highly}. Hence, according to Corollary \ref{cor:Lctqw_pbm}, all the graphs whose trap is the fully connected vertex $w$ have $\eta=\eta_K$. This follows from the fact that $\eta$ is the overlap of the initial state with the basis states of the invariant subspace $\mathcal{I}(H,\vert w \rangle )$ \cite{novo2015systematic,razzoli2021transport},
\begin{equation}
\eta = \sum_{n=1,2} \vert \langle e_n \vert \psi_0 \rangle \vert^2 = \vert \langle e_2 \vert \psi_0 \rangle \vert^2 =\frac{1}{N-1}\,,
\end{equation}
and such invariant subspace is \eqref{eq:basis_inv_subsp_G} for the problems and graphs under investigation, including the complete graph.

Alternatively, we can prove this as follows. We define the integrated probability of trapping within the time interval $[0,t]$,
\begin{eqnarray}
\tilde{\eta}(t) = 2 \kappa \int_0^{t} \langle w \vert  \rho(\tau) \vert w \rangle \,d\tau \quad \Rightarrow \quad \lim_{t \to +\infty}\tilde{\eta}(t) = \eta\,,
\end{eqnarray}
where $\langle w \vert \rho(t) \vert w \rangle = \vert \langle w \vert \exp\left[-i H t\right] \vert v \rangle\vert ^2$. From Equation \eqref{eq:prob_ampli_target}, the probability amplitude at $w$,
\begin{eqnarray}
\langle w \vert e^{-iHt}\vert v \rangle &=\langle w \vert e^{-i H_{\mathrm{red}}t}\sum_{n=1,2} \vert e_n \rangle \langle e_n \vert v \rangle = \frac{1}{\sqrt{N-1}}\langle e_1 \vert e^{-i H_{\mathrm{red}}t}\vert e_2 \rangle\,,
\end{eqnarray}
is independent (i) of the graph under investigation and (ii) of the initial vertex state $\vert v \rangle$, provided that $v \neq w$. (i) Follows from the fact that the graphs considered have the same basis states and the same reduced Hamiltonian (Corollary \ref{cor:Lctqw_pbm}). (ii) Follows from the fact that all the vertices other than the trap only overlap with $\vert e_2 \rangle$, which is the equal superposition of them, and have the same overlap with it.
Therefore, $\tilde{\eta}(t)$ does not depend on the graph under investigation or on the initial vertex state. As a result, in the limit of infinite time we also recover the same transport efficiency $\eta = \eta_K$.

In this problem the initial state is a vertex state $\vert v \neq w \rangle$ and cannot be written as linear combination of the two basis states. Therefore, it evolves differently depending on the given graph. Nevertheless, as just shown, it provides the same dynamics relevant to the problem, i.e., the same (trapped) population at $w$. 

\section{Grover-like CTQWs with multiple marked vertices}
\label{sec:equiv_CTQW_pbm_multiple}

\begin{theorem}
\label{th:Lctqw_pbm_multiple}
Let $G=(V,E)$ be a simple graph of order $N=\vert V \vert$ with $M = \vert E \vert$ edges. Let $W := \left\lbrace v \in V \mid \mathrm{deg}(v)=N-1 \wedge \textrm{$v$ is marked}\right\rbrace \neq \emptyset$ be the set of fully connected marked vertices and let $\mu:= \vert W \vert$, with $1 \leq \mu < N$.
 Let us consider a Grover-like CTQW  where the quantities of interest are the probability amplitudes at $w \in W$. Let
\begin{equation}
H = \gamma L + \sum_{w \in W} \lambda_w \vert w \rangle\langle w \vert
\label{eq:Ham_pbm_multiple}
\end{equation}
be the Hamiltonian encoding the problem, where $\gamma \in \mathbb{R}^+$ is constant and $\lambda_w \in \mathbb{C}$ depends on the fully connected vertex.  
Then, given the initial state $\vert \psi_0 \rangle$, the probability amplitude at a marked vertex is $\langle w \vert \exp\left[ -i H_{\mathrm{red}} t\right] \vert {\psi_0}_\mathrm{red} \rangle$, where $\vert{\psi_0}_\mathrm{red} \rangle = P\vert \psi_0 \rangle$ with $P$ the projector onto the $(\mu+1)$-dimensional invariant subspace relevant to the problem,
\begin{equation}
\mathcal{I} = \mathrm{span}\left(\left\lbrace \{\vert e_k \rangle =\vert w_k \rangle \}_k, \vert e_{\mu+1} \rangle =\frac{1}{\sqrt{N-\mu}}\sum\nolimits_{v \notin W} \vert v \rangle \right\rbrace\right)\,,
\label{eq:basis_inv_subsp_multiple}
\end{equation}
with $k = 1,\ldots,\mu$, and the reduced Hamiltonian is
\begin{equation}
H_\mathrm{red} = \gamma
\left( \begin{array}{ccccc}
\Delta_1+\lambda_{w_1}' & -1 	& \cdots &-1 & -\sqrt{\Delta_\mu}\\
-1						& \ddots& \ddots &\vdots & \vdots\\
\vdots & \ddots & \ddots & -1 & \vdots\\
-1					& \cdots & -1& \Delta_1+\lambda_{w_{\mu}}'	& -\sqrt{\Delta_\mu}\\
-\sqrt{\Delta_\mu} 		& \cdots & \cdots & -\sqrt{\Delta_\mu} 	&			 \mu
\end{array} \right)\,,
\label{eq:Ham_pbm_multiple_red}
\end{equation}
where $\Delta_n = N-n$  and $\lambda_{w}'= \lambda_{w}/\gamma$.
Grover-like CTQWs on a graph $G_1$ and on a graph $G_2$ both of order $N$ result in the same probability amplitude $\langle w \vert \exp\left[ -i H_{G_1} t\right] \vert \psi_{0,G_1} \rangle=\langle w \vert \exp\left[ -i H_{G_2} t\right] \vert \psi_{0,G_2} \rangle$ provided that $\vert{\psi_{0,G_1}}_\mathrm{red} \rangle = \vert{\psi_{0,G_2}}_\mathrm{red} \rangle$.
\end{theorem}

\begin{remark}
\label{rem:identically_evol_v}
The dimensionality of the problem can be further reduced if subsets of vertices in $W$ have the same $\lambda$, $W_\alpha =\left\lbrace w \in W \mid \lambda_w = \lambda_\alpha \right\rbrace$ such that $\bigcup_{\alpha}W_\alpha = W$ and $W_\alpha \cap W_\beta = \emptyset$ $\forall \alpha \neq \beta$. Instead of having one basis state per marked vertex, the equal superposition of all vertex states from the same set $W_\alpha$ defines one basis state, $\vert e_{W_\alpha} \rangle=\sum_{w \in W_\alpha}\vert w \rangle /\sqrt{\vert W_\alpha\vert}$. This follows from the symmetries of the problem, as they allow to group together identically evolving vertices \cite{wang2021role}. The reduced Hamiltonian \eqref{eq:Ham_pbm_multiple_red} will change according to the new basis.
\end{remark}

\begin{proof}
We have more than one marked vertex and we cannot apply straightforwardly the dimensionality reduction method, because  neither the initial state is unique (except in the spatial search) nor the target state is unique (multiple marked vertices). The Hamiltonian \eqref{eq:Ham_pbm_multiple} inherits the symmetries of the graph (Laplacian matrix), but each oracle Hamiltonian $H_w$ breaks the symmetries involving the corresponding fully connected vertex $w$. 
Here we consider the Hamiltonian in the general framework, with no assumptions on $\lambda$'s.

During the time evolution of the system the population at the marked vertices is determined only by the Hamiltonian eigenstates having nonzero overlap with the marked vertices. Our aim is to prove that the subspace $\mathcal{E}$ spanned by those eigenstates is the subspace $\mathcal{I}$ \eqref{eq:basis_inv_subsp_multiple}. Let us define the subspace
\begin{equation}
    \mathcal{E} := \mathrm{span}\left(\left\lbrace \vert \varepsilon \rangle \mid H \vert \varepsilon \rangle = \varepsilon \vert \varepsilon\rangle \wedge \langle w \in W \vert \varepsilon \rangle\neq 0 \right\rbrace\right)\,,
\end{equation}
where the $\vert \varepsilon \rangle $ are the minimum number of Hamiltonian eigenstates overlapping with the fully connected marked vertices $w \in W$. By mininum we mean that in the case of degenerate eigenspaces more than one eigenstate can have a nonzero overlap with the marked vertices. We can solve this ambiguity by choosing the eigenstate from this degenerate eigenspace which has the maximum possible overlap with the marked vertices and then by orthogonalizing all the other vectors within this eigenspace with respect to it. Therefore, after orthogonalization, the remaining eigenstates in the degenerate space would have zero overlap with the marked vertices. This approach to the problem is explained in \cite{caruso2009highly}, where it provides a simple way to compute the efficiency of transport to a trapping vertex on a graph (in the absence of dephasing and dissipation).

\begin{lemma}
\label{lemma:Eps_notin_sum}
The Hamiltonian eigenstates that do not overlap with the marked vertices have projections onto the vertex states that sum to zero, 
\begin{equation}
\sum_{v} \langle v \vert \varepsilon\notin \mathcal{E} \rangle = \sum_{v\notin W} \langle v \vert \varepsilon\notin \mathcal{E} \rangle =0\,.    
\end{equation}
\end{lemma}

\begin{subproof}
We study the eigenproblem $H \vert \varepsilon \rangle = \varepsilon \vert \varepsilon \rangle$ by components in the basis of vertex states, projecting the eigenvalue equation onto a generic $\vert v \rangle$
\begin{eqnarray}
    \langle v \vert H \vert \varepsilon \rangle - \varepsilon \langle v \vert \varepsilon \rangle
    &=\sum_{v'} \left[ \gamma (D_{vv'}-A_{vv'})\right] \langle v' \vert \varepsilon \rangle+\sum_{w \in W }\lambda_w\langle v \vert w \rangle\langle w \vert \varepsilon \rangle- \varepsilon \langle v \vert \varepsilon \rangle\nonumber\\
    &=(\gamma d_v - \varepsilon+\lambda_w \delta_{vw})\langle v \vert \varepsilon \rangle -\gamma \sum_{v'} A_{vv'}\langle v' \vert \varepsilon \rangle= 0\,.
    \label{eq:eigpbm_onto_v}
\end{eqnarray}
Let us focus on $\vert \varepsilon \rangle \notin \mathcal{E}$ and $v \in W$. Then, from Equation \eqref{eq:eigpbm_onto_v}, we have
\begin{equation}
  \sum_{v' \neq v} \langle v' \vert \varepsilon \rangle
  =\sum_{v'} \langle v' \vert \varepsilon \rangle
  =\sum_{v' \notin W} \langle v' \vert \varepsilon \rangle = 0\,,
\end{equation}
as $v\in W$ is fully connected, thus $A_{vv'}=1$ $\forall v \neq v'$ ($A_{vv}=0$). The index of summation can be extended to all the vertices $v' \in V$ or  limited to $v' \notin W$ as $\langle v' \in W \vert \varepsilon \notin \mathcal{E} \rangle=0$ by definition.
\end{subproof}

\begin{lemma}
\label{lemma:v_notin_Eps_in_const}
The Hamiltonian eigenstates that overlap with the marked vertices have constant projection onto the non-marked vertex states, 
\begin{equation}
    \langle v \notin W  \vert \varepsilon \in \mathcal{E}\rangle = \frac{\gamma}{\gamma \mu -\varepsilon}\sum_{v' \in W} \langle v' \vert \varepsilon \rangle=const \quad \forall v \notin W\,.
\end{equation}
\end{lemma}
\begin{subproof}
From Equation \eqref{eq:eigpbm_onto_v}, the components under investigation are
\begin{equation}
    \langle v \notin W \vert \varepsilon \in \mathcal{E}\rangle = \frac{\gamma \sum_{v'} A_{vv'} \langle v' \vert \varepsilon \rangle}{\gamma d_v-\varepsilon}
    = \frac{\gamma}{\gamma d_v-\varepsilon}\left(\xi + \sum_{v'\notin W} A_{vv'} \langle v' \vert \varepsilon \rangle\right)\,,
   \label{eq:v_notin_eps_in_step0}
\end{equation}
where we have defined $\xi := \sum_{v' \in W} \langle v' \vert \varepsilon \rangle$, which does not depend on the $v \notin W$ chosen, and we have used $A_{vv'}=1$  $\forall v' \in W$. Indeed, $v \notin W$, thus $v \neq v'$, and the vertices $v'$ are the fully connected ones.

Let us start with a particular case. If the vertices $v \notin W$ are only connected to the vertices $w \in W$, then 
$d_v = \mu = \vert W \vert$ $\forall v \notin W$ and $A_{vv'} = 0$ $\forall v,v' \notin W$. Hence, all the components are constant and equal to 
\begin{equation}
    \langle v \notin W  \vert \varepsilon \in \mathcal{E}\rangle = \frac{\gamma \xi}{\gamma \mu -\varepsilon}\quad \forall v \notin W\,.
    \label{eq:v_notin_eps_in_StarLike}
\end{equation}
In general, instead, we have a system of $\bar{\mu}:= N-\mu$ linear equations like \eqref{eq:v_notin_eps_in_step0} in $\bar{\mu}$ unknowns $x_j := \langle v_j \notin W \vert \varepsilon \in \mathcal{E} \rangle$, with $1 \leq j \leq \bar{\mu}$,
\begin{equation}
    {\everymath={\displaystyle}
    \left\lbrace
    \begin{array}{rcl}
         x_1 - \frac{\gamma}{\gamma d_1-\varepsilon}\sum_{k \neq 1} A_{1k } x_k &=& \frac{\gamma \xi}{\gamma d_1-\varepsilon}\\
         &\vdots \\
         x_{\bar{\mu}}  - \frac{\gamma}{\gamma d_{\bar{\mu}}-\varepsilon}\sum_{k \neq \bar{\mu}} A_{\bar{\mu} k } x_{k}&=& \frac{\gamma\xi}{\gamma d_{\bar{\mu}}-\varepsilon} \,.
    \end{array}
    \right.
    }
    \label{eq:sys_lin_eq}
\end{equation}
We make the following ansatz on the solution
\begin{equation}
    x_1 = \ldots = x_{\bar{\mu}} = \frac{\gamma \xi}{\gamma \mu -\varepsilon}\,,
    \label{eq:ansatz_const_comp}
\end{equation}
based on the analytical solution  \eqref{eq:v_notin_eps_in_StarLike} for a particular case and on numerical evidence for general graphs, including the complete graph. Hence, focusing on the left-hand side of the $j$-th equation \eqref{eq:sys_lin_eq}, we recover the identity with the right-hand side of the same equation
\begin{eqnarray}
\frac{\gamma \xi}{\gamma \mu -\varepsilon} \left[ 1- \frac{\gamma}{\gamma d_j-\varepsilon}\sum_{k \neq j} A_{jk}\right] &=
\frac{\gamma \xi}{\gamma \mu -\varepsilon} \left[1- \frac{\gamma}{\gamma d_j-\varepsilon}(d_{j}-\mu)\right]\nonumber\\
&=\frac{\gamma \xi}{\gamma \mu -\varepsilon} \frac{\gamma \mu -\varepsilon}{\gamma d_j-\varepsilon} = \frac{\gamma \xi}{\gamma d_j-\varepsilon}\,,
\end{eqnarray}
where $\sum_{k\neq j} A_{jk}= d_j - \mu$ because the index of summation does not run over all the vertices but runs over the non-marked vertices, hence we get the degree $d_j$ lowered by the number of fully connected marked vertices, $\mu$. This identity applies to all $j=1,\ldots, N-\mu$, i.e., to all $v \notin W$. This verifies the correctness of the ansatz \eqref{eq:ansatz_const_comp} and therefore proves the Lemma.
\end{subproof}

According to the previous Lemmas, we now prove that $\mathcal{E} = \mathcal{I}$. First, we prove that $\mathcal{E} \subseteq \mathcal{I}$. Let $c:=\langle v \notin W \vert \varepsilon \in \mathcal{E}\rangle$ (Lemma \ref{lemma:v_notin_Eps_in_const}). Then, we can write any $\vert \varepsilon \rangle \in \mathcal{E}$ as
\begin{eqnarray}
    \vert \varepsilon \rangle &= \sum_v \vert v \rangle\langle v \vert \varepsilon \rangle = \sum_{w \in W} \vert w \rangle\langle w \vert \varepsilon \rangle + c \sum_{v \notin W} \vert v \rangle\nonumber\\
    &= \sum_{n=1}^\mu \vert e_n \rangle\langle e_n \vert \varepsilon \rangle + c \sqrt{N-\mu} \vert e_{\mu+1} \rangle \in \mathcal{I}\,,
\end{eqnarray}
as it is a linear combination of the basis states \eqref{eq:basis_inv_subsp_multiple}. Second, we prove that $\mathcal{I} \subseteq \mathcal{E}$. We start with the basis states $\vert e_j\rangle = \vert w_j \rangle$ for $j=1,\ldots, \mu$
\begin{equation}
    \vert e_j \rangle = \sum_{\varepsilon} \vert \varepsilon \rangle\langle \varepsilon \vert e_j \rangle
    = \sum_{\varepsilon \in \mathcal{E}} \vert \varepsilon \rangle\langle \varepsilon \vert e_j \rangle \in \mathcal{E}\,,
\end{equation}
as it is a linear combination of the Hamiltonian eigenstates $\vert \varepsilon \rangle \in \mathcal{E}$. The summation over $\varepsilon$ denotes the summation over all the Hamiltonian eigenstates. The second equality follows from $\langle w_j \vert \varepsilon \notin \mathcal{E} \rangle = 0$, by definition. The last basis state is 
\begin{eqnarray}
    \vert e_{\mu + 1} \rangle &= \frac{1}{\sqrt{N-\mu}}\sum_{\varepsilon} \sum_{v\notin W} \vert \varepsilon \rangle\langle \varepsilon \vert v \rangle\nonumber\\
    &= \frac{1}{\sqrt{N-\mu}}\left[ \sum_{\varepsilon \in \mathcal{E}} \sum_{v\notin W} \vert \varepsilon \rangle\langle \varepsilon \vert v \rangle + \sum_{\varepsilon \notin \mathcal{E}} \sum_{v\notin W} \vert \varepsilon \rangle\langle \varepsilon \vert v \rangle\right]\nonumber\\
    &= \frac{1}{\sqrt{N-\mu}}\left[ \sum_{\varepsilon \in \mathcal{E}} \sum_{v\notin W} \vert \varepsilon \rangle\langle \varepsilon \vert v \rangle +0 \right] \in \mathcal{E}\,,
\end{eqnarray}
where the last equality follows from Lemma \ref{lemma:Eps_notin_sum}.
To summarize, $\mathcal{I} = \mathcal{E}$ \eqref{eq:basis_inv_subsp_multiple}, since $\mathcal{E} \subseteq \mathcal{I}$ and $\mathcal{I} \subseteq \mathcal{E}$, and this also implies that $\dim \mathcal{E} =\dim \mathcal{I} = \mu+1$.

Now that we have the basis of the invariant subspace, we can write the reduced Hamiltonian. Given the Hamiltonian \eqref{eq:Ham_pbm_multiple}, the matrix elements of the reduced Hamiltonian for $j,k=1,\ldots,\mu$ are
\begin{eqnarray}
\langle e_j \vert  H \vert e_k \rangle  &= \gamma \left( \langle e_j \vert  D \vert e_k \rangle -\langle e_j \vert  A \vert e_k \rangle \right)+ \lambda_{w_j} \delta_{jk}\nonumber\\
&= \left[\gamma (N-1)+\lambda_{w_j}\right] \delta_{jk}-\gamma\,,
\end{eqnarray}
since $d_{w}=N-1$ and the vertices $w_j$ and $w_k$ are necessarily adjacent,
\begin{eqnarray}
\langle e_j \vert  H \vert e_{\mu+1} \rangle  &= -\gamma \langle e_j \vert  A \vert e_{\mu+1} \rangle
= \frac{-\gamma}{\sqrt{N-\mu}} \sum_{v \notin W} A_{w_j v}
= -\gamma \sqrt{N-\mu}\,,
\end{eqnarray}
since the basis is orthonormal and $A_{w_j v}=1$ $\forall v \notin W \wedge \forall w \in W$ ($w$ is fully connected).
The last element is
\begin{equation}
\langle e_{\mu+1} \vert  H \vert e_{\mu+1} \rangle  = \frac{\gamma}{N-\mu}\sum_{v,v'\notin W} L_{vv'}= \gamma \mu\,.
\end{equation}
Indeed,
\begin{equation}
\sum_{v\notin W} d_v = \sum_{v\in V} d_v-\mu(N-1)=2M-\mu(N-1)\,,
\end{equation}
and
\begin{eqnarray}
\sum_{v,v'\notin W} A_{vv'} &= \sum_{v,v' \in V}A_{vv'}-\sum_{v\in W}\sum_{v'\in V} A_{vv'}-\sum_{v\notin W}\sum_{v'\in W} A_{vv'}\nonumber\\
&= \sum_{v,v' \in V}A_{vv'}-\sum_{v\in W} d_{v}-\sum_{v\notin W}\mu\nonumber\\
&= 2M-\mu(N-1)-(N-\mu)\mu\,,
\end{eqnarray}
since $A_{vv'}=1$ $\forall v' \neq v \wedge v \in W$ and, we recall, $\mu=\vert W \vert$ and $N=\vert V \vert$. Hence, the reduced Hamiltonian \eqref{eq:Ham_pbm_multiple_red} follows.
\end{proof}

\subsection{Spatial search}
The Hamiltonian encoding the problem is
\begin{equation}
    H = \gamma L - \sum_{w \in W}\vert w \rangle \langle w \vert\,,
    \label{eq:Ham_search_multiple}
\end{equation}
where the marked vertices, the $\mu$ possible solutions of the spatial search, are the fully connected vertices $w\in W$. The oracles are unbiased, $\lambda_w = -1$ $\forall w \in W$ in Equation \eqref{eq:Ham_pbm_multiple}, as the solutions are usually assumed to be equivalent \cite{roland2002quantumsearch,wong2016spatial}. The goal is to tune the hopping amplitude $\gamma$ to maximize success probability $P_W(t) = \sum_{w \in W} P_w(t)$ after a period of time of evolution. The overall success probability $P_W$ is the sum of the probabilities at each $w \in W$ because these are equivalent solutions. Solving the problem amounts to finding one of them. The initial value is $P_W(0) = \mu/N$, since the initial state is the equal superposition of all vertices. The time evolution of $\vert \psi_0 \rangle$ is entirely contained in $\mathcal{I}$, as $\vert \psi_0 \rangle = (\sum_{j=1}^{\mu} \vert e_j \rangle +\sqrt{N-\mu} \vert e_{\mu+1} \rangle )/\sqrt{N}$ and so $\vert {\psi_0}_{\mathrm{red}} \rangle = \vert {\psi_0} \rangle$. Hence, not only the success probability $P_W(t)$, but also the entire dynamics of the system $\exp\left[-i H t\right] \vert {\psi_0} \rangle$ is the same on any simple graph $G$. 
According to Theorem \ref{th:Lctqw_pbm_multiple}, the results we have for the spatial search on the complete graph also apply to the search of $w\in W$ on other graphs. The spatial search of $\mu$ marked vertices in the complete graph is known to be optimal ($P_W=1$) for $\gamma = 1/N$ at time $t^\ast = (\pi/2)\sqrt{N/\mu}$ \cite{wong2016spatial}. We point out that in \cite{wong2016spatial} the CTQW is generated by the adjacency matrix, but this is equivalent to using the Laplacian matrix since the complete graph is regular. Hereafter we prove these results on the optimal search without assuming that the graph is complete.

Spatial search is a suitable case study to apply Remark \ref{rem:identically_evol_v}, as all the fully connected marked vertices have the same $\lambda = -1$. Therefore, the Hamiltonian \eqref{eq:Ham_search_multiple} is invariant under permutations of the vertices in $W$. This symmetry allows us to further reduce the dimensionality of the problem by grouping together such identically evolving vertices in the state $\vert \tilde{e}_1 \rangle=\sum_{w \in W}\vert w \rangle /\sqrt{\mu}$ \cite{wang2021role}. This state is the solution of the search and is the first basis state of the reduced invariant subspace. Then, it can be shown that
\begin{eqnarray}
H \vert \tilde{e}_1 \rangle  &= [\gamma(N-\mu)-1]\vert \tilde{e}_1 \rangle -\gamma \sqrt{\mu(N-\mu)} \vert \tilde{e}_2 \rangle\,,\\
H\vert \tilde{e}_2 \rangle &= -\gamma \sqrt{\mu(N-\mu)} \vert \tilde{e}_1 \rangle+ \gamma\mu\vert \tilde{e}_2 \rangle\,,
\end{eqnarray}
where $\vert \tilde{e}_2 \rangle :=\sum_{v \notin W} \vert v \rangle/\sqrt{N-\mu}$ is the second basis state. Therefore, the orthonormal states $\vert \tilde{e}_1 \rangle$ and $\vert \tilde{e}_2 \rangle$ span the invariant subspace relevant to the spatial search. The reduced Hamiltonian is
\begin{equation}
H_\mathrm{red} = \gamma
\left( \begin{array}{cc}
N-\mu-1/\gamma & -\sqrt{\mu(N-\mu)}\\
-\sqrt{\mu(N-\mu)} & \mu
\end{array} \right)\,.
\label{eq:Ham_spatial_multiple_red}
\end{equation}
For $\gamma = 1/N$, the eigenvalues are $\varepsilon_{\pm} = \pm \sqrt{\mu/N}$ and the corresponding eigenstates are
\begin{equation}
    \vert \varepsilon_\pm \rangle = \sqrt{\frac{\sqrt{N}\pm\sqrt{\mu}}{2 \sqrt{N}}}\left(\mp\frac{\sqrt{N-\mu}}{\sqrt{N}\pm\sqrt{\mu}} \vert \tilde{e}_1 \rangle + \vert \tilde{e}_2 \rangle\right)\,.
\end{equation}
The success probability
\begin{equation}
    P_W(t) = \vert \langle \tilde{e}_1 \vert e^{-i H t}  \vert \psi_0 \rangle \vert^2
    = \frac{\mu}{N} \cos^2\left(\sqrt{\frac{\mu}{N}}t\right)+\sin^2\left(\sqrt{\frac{\mu}{N}}t\right) \,,
\end{equation}
is equal to one (certainty) at time $t^\ast = (\pi/2)\sqrt{N/\mu}$. 
For $\mu = 1$ we recover the results---reduced Hamiltonian, success probability, and optimal time---for the spatial search of a single marked vertex discussed in Section \ref{subsec:spatial_search_single}.

\subsection{Quantum transport}
The non-Hermitian effective Hamiltonian encoding the problem is
\begin{equation}
    H = L - i \sum_{w \in W} \kappa_w \vert w \rangle \langle w \vert\,,
    \label{eq:Ham_transport_multiple}
\end{equation}
where the $\mu$ trapping vertices are the fully connected vertices $w\in W$ and have, in general, different trapping rates $\kappa_w\in \mathbb{R}^+$ ($\lambda_w = -i\kappa_w$ in \eqref{eq:Ham_pbm_multiple}). Accordingly, $\eta := 2 \sum_{w \in W}\kappa_w \int_0^{+\infty} \langle w \vert  \rho(t) \vert w \rangle \,dt$ \cite{rebentrost2009environment}. We assume $\vert \psi_0 \rangle=\vert v \notin W\rangle $, therefore, according to the basis states \eqref{eq:basis_inv_subsp_multiple},
\begin{equation}
\eta = \sum_{n=1}^{\mu+1} \vert \langle e_n \vert \psi_0 \rangle \vert^2 = \vert \langle e_{\mu+1} \vert \psi_0 \rangle \vert^2 = \frac{1}{N-\mu}\,.
\label{eq:transport_eff_multiple_fcv}
\end{equation}
The transport efficiency improves as the number of fully connected traps $\mu$ increases and does not depend on the trapping rates. Changing the $\kappa_w$ affects the timescale on which the trapping occurs, not $\eta$ as it is defined in the limit of infinite time. Moreover, $\tilde{\eta}(t)= 2 \sum_{w \in W}\kappa_w \int_0^{t} \langle w \vert  \rho(\tau) \vert w \rangle \,d\tau$ does not depend on the initial vertex state $\vert v \rangle $, provided that $v \notin W$. Indeed, from Equation \eqref{eq:prob_ampli_target}, the probability amplitude at $w \in W$,
\begin{eqnarray}
\langle w \vert e^{-iHt}\vert v \rangle &=\langle w \vert e^{-i H_{\mathrm{red}}t}\sum_{n=1}^{\mu+1} \vert e_n \rangle \langle e_n \vert v \rangle = \frac{1}{\sqrt{N-\mu}}\langle e_{w} \vert e^{-i H_{\mathrm{red}}t}\vert e_{\mu+1} \rangle\,,
\end{eqnarray}
is independent of $v \notin W$. For $\mu = 1$ we recover the transport efficiency for the single trapping vertex discussed in Section \ref{subsec:quantum_transport_single}.

\section{Conclusions}
\label{sec:conclusion}
In this paper we have investigated the role of the fully connected vertex $w$ in continuous-time quantum walks (CTQWs) on simple graphs $G$ of order $N$. In particular, we have analytically proved that when the dynamics of the walker is governed by the Laplacian matrix, the CTQW starting from the state $\vert w \rangle $ does not depend on the graph $G$ considered and it is therefore equivalent, e.g., to the CTQW on the complete graph of the same order, $K_N$. Instead, the corresponding adjacency CTQWs do depend on the graph considered. 

After that, we have investigated Grover-like CTQWs, i.e., systems with  Hamiltonian of the form $H = \gamma L +\sum_{w \in W}\lambda_w \vert w \rangle\langle w \vert $, where
$W$ is the subset of vertices made of $\mu$ fully connected marked vertices. Here the quantity of interest is the probability amplitude at the vertices $w \in W$.
For these systems, we have analytically proved that the probability amplitudes of interest do not depend on the graph considered. In this case,  the equivalence  concerns the dynamics relevant to the computation of the probability amplitude at $w$, whereas the full dynamics of the walkers are not necessarily equivalent.

As applications of the above results, we have considered spatial search of $w \in W$ and quantum transport to $w \in W$. These problems on a simple graph $G$ of order $N$ inherit the results already known for the corresponding problems on the complete graph $K_N$, independently of the considered graph. In particular, the spatial search of equivalent solutions (unbiased oracles) is optimal for $\gamma = 1/N$ at time $t^\ast = (\pi/2)\sqrt{N/\mu}$, and the full dynamics of the equal superposition of all vertices under the search Hamiltonian on $G$ and on $K_N$ are equivalent. Regarding quantum transport of an initially localized excitation, the transport efficiency $\eta$ increases with the number of fully connected traps as $\eta = 1/(N-\mu)$, and does not depend on the initial vertex state $\vert v \notin W\rangle$. 

Our proofs are based on the notion of Krylov subspaces. We have determined the invariant subspace relevant to the considered Laplacian problems,  and the corresponding reduced Hamiltonian, thus reducing the dimensionality of the original problem. Whenever a fully connected vertex is the initial state of the CTQW or a marked vertex of a Grover-like CTQW, results do not depend on the graph considered. Hence, the universality of the fully connected vertex.

One of most relevant consequences of our work is that the spatial search of fully connected vertices is always optimal and does not depend on the full topology of the involved graph. We can always find the solution with certainty and we know the parameters, $\gamma$ and time, to achieve this result. This can be exploited, e.g., in finding the fully connected hubs of a network. Indeed, most often the hub is not connected to all the nodes, but serves as the center of star-shaped subnetwork \cite{sakarya2020hybrid} and our results hold when applied to the subnetwork. More generally, our results provide a coherent and unified framework 
to understand and extend several partial results already reported in literature for fully connected vertices, and pave the way for further development in the area, e.g., understanding whether universality survives in the presence of chirality \cite{chi1,chi2}.

\ack
Work done under the auspices of GNFM-INdAM. The authors thank Claudia Benedetti and Massimo Frigerio for helpful discussions.

\section*{References}
\bibliographystyle{iopart-num-mod}
\bibliography{biblio_univ_fcv_Lctqw_pbm}

\end{document}